\documentclass[10pt,letterpaper,journal]{IEEEtran}

\usepackage{amsmath,amssymb,amsthm,mathtools,bm}
\usepackage{newtxtext,newtxmath}
\usepackage{microtype}
\usepackage{flushend}
\usepackage{cite}
\usepackage{url}
\usepackage{hyperref}
\usepackage[nameinlink,capitalize]{cleveref}

\newcommand{\R}{\mathbb{R}}
\newcommand{\T}{\mathbb{T}}

\newcommand{\Z}{\mathbb{Z}}

\newcommand{\cH}{\mathcal{H}}
\newcommand{\cE}{\mathcal{E}}

\newcommand{\dd}{\,\mathrm{d}}
\newcommand{\ii}{\mathrm{i}}

\newcommand{\avgT}[1]{\left\langle #1\right\rangle_{\T}}
\newcommand{\avgG}[2]{\left\langle #1\right\rangle_{#2}}

\theoremstyle{plain}
\newtheorem{theorem}{Theorem}

\newtheorem{lemma}[theorem]{Lemma}
\newtheorem{corollary}[theorem]{Corollary}
\theoremstyle{definition}

\newcommand{\finalcolumnbalance}{}

\hypersetup{
  colorlinks=true,
  linkcolor=black,
  citecolor=black,
  urlcolor=black,
  pdfauthor={Jiayang Zou, Luyao Fan, Jiayang Gao, and Jia Wang},
  pdftitle={An Explicit Family of Log-Concave Counterexamples to the Gaussian Completely Monotone Conjecture},
  pdfsubject={Gaussian complete monotonicity under log-concavity},
  pdfkeywords={complete monotonicity, entropy, Fisher information, heat flow, log-concavity}
}

\title{An Explicit Family of Log-Concave Counterexamples to the Gaussian Completely Monotone Conjecture}

\author{Jiayang Zou\textsuperscript{1,2}, Luyao Fan\textsuperscript{2},
Jiayang Gao\textsuperscript{2}, and Jia Wang\textsuperscript{2}\\[0.5ex]
\small \textsuperscript{1}Stanford University, Stanford, CA, USA;
\texttt{jyangzou@stanford.edu}\\
\small \textsuperscript{2}Shanghai Jiao Tong University, Shanghai, China;
\texttt{\char123 qiudao, fanluyao, gjy0515, jiawang\char125 @sjtu.edu.cn}}
 
\begin{document}
\maketitle

% Keep display mathematics compact in the final two-column layout.
\setlength{\abovedisplayskip}{4pt plus 2pt minus 1pt}
\setlength{\belowdisplayskip}{4pt plus 2pt minus 1pt}
\setlength{\abovedisplayshortskip}{1pt plus 1pt}
\setlength{\belowdisplayshortskip}{3pt plus 2pt minus 1pt}

\begin{abstract}
We construct smooth, strictly log-concave counterexamples to the Gaussian
completely monotone conjecture in every dimension. In one dimension, they form
an explicit family $f_m$ whose signed $m$th entropy derivative at time zero is
negative for every sufficiently large $m$; the inequality persists for all
sufficiently small positive times. Tensorization with a broad Gaussian factor
gives the higher-dimensional examples. The argument is analytic and
self-contained. It reduces the sign to a two-frequency entropy calculation on
the circle and transfers the resulting asymptotic to the real line through an
exact heat-flow formula for Gaussian-windowed Fourier modes. The proof was
developed by GPT-5.6 Sol Pro under the authors' guidance.
\end{abstract}

% \begin{IEEEkeywords}
% Complete monotonicity, entropy, Gaussian convolution, heat flow,
% log-concavity.
% \end{IEEEkeywords}

\section{Introduction}
\label{sec:introduction}

Let $P_t=e^{t\partial_x^2/2}$ denote the heat semigroup on $\R$.  If $X$
has density $f$, then $P_tf$ is the density of $X+\sqrt{t}\,Z$, where $Z$
is a standard Gaussian independent of $X$.  We consider the Boltzmann
entropy profile
\begin{equation}
 \cH_f(t)=\int_{\R}(P_tf)(x)\log(P_tf)(x)\,\dd x.
 \label{eq:entropy-profile}
\end{equation}
The Gaussian completely monotone conjecture of Cheng and Geng
\cite{ChengGeng2015} asserts that
\begin{equation}
 (-1)^k\cH_f^{(k)}(t)\ge0,
 \qquad k\ge1,\quad t>0.
 \label{eq:gcmc}
\end{equation}
Our heat convention corresponds to adding Gaussian noise of variance $t$;
the alternative convention $e^{t\partial_x^2}$ only rescales time.

We begin with the construction.  For each integer $m\ge5$, set
\begin{equation}
 a_m=2^{-m/4},\qquad b_m=2^{-m},\qquad \kappa_m=8a_m,
 \label{eq:parameters}
\end{equation}
and consider the probability density
\begin{equation}
 f_m(x)=Z_m^{-1}e^{-\kappa_mx^2/2}
 \bigl(1+a_m\cos x+b_m\cos2x\bigr),
 \qquad x\in\R,
 \label{eq:candidate}
\end{equation}
where $Z_m$ is the normalizing constant.

\begin{theorem}
\label{thm:main}
For every $m\ge5$, the density $f_m$ in~\eqref{eq:candidate} is smooth,
strictly positive, Gaussian-decaying, and strictly log-concave. Moreover,
with
\begin{equation}
 \beta:=\frac{3}{2\sqrt2}>1,
 \label{eq:beta}
\end{equation}
one has
\begin{equation}
 (-1)^m\cH_{f_m}^{(m)}(0)
 =-\frac{11}{96}\beta^m+o(\beta^m)
 \qquad (m\to\infty).
 \label{eq:main-asymptotic}
\end{equation}
Consequently, for every sufficiently large $m$, there is
$\varepsilon_m>0$ such that
\begin{equation}
 (-1)^m\cH_{f_m}^{(m)}(t)<0,
 \qquad 0\le t<\varepsilon_m.
 \label{eq:positive-time-failure}
\end{equation}
Thus~\eqref{eq:gcmc} fails for smooth strictly log-concave densities on~$\R$.
\end{theorem}

Corollary~\ref{cor:high-dimensional} tensors $f_m$ with an explicit broad
Gaussian factor and gives the same asymptotic failure in every dimension.

The conjecture belongs to a hierarchy of Gaussian inequalities arising from
entropy power along the heat flow. Costa proved the concavity of entropy power
under Gaussian convolution \cite{Costa1985,Villani2000}. McKean's Gaussian
optimality conjecture asks whether the signed entropy derivatives are
extremized by Gaussian distributions of the same variance
\cite{McKean1966}, while Toscani proposed a third-order entropy-power inequality
that helped motivate this hierarchy \cite{Toscani2015}.
Wang proved that the entropy power conjecture implies McKean's conjecture
\cite{Wang2024Implication}, and McKean's conjecture in turn
implies~\eqref{eq:gcmc}. Thus
\[
 \begin{gathered}
  \text{entropy-power conjecture}
  \;\Longrightarrow\; \text{McKean's conjecture},\\[-0.2ex]
  \text{McKean's conjecture}
  \;\Longrightarrow\; \text{GCMC}.
 \end{gathered}
\]

Several lower-order cases of~\eqref{eq:gcmc} were known before the recent
counterexamples.  Cheng
and Geng proved~\eqref{eq:gcmc} in one dimension through order four
\cite{ChengGeng2015}.  In arbitrary dimension the first two orders follow
from entropy-power concavity, and the third-order inequality is known in
dimensions at most four \cite{GuoYuanGao2022}.  For one-dimensional
log-concave inputs, Zhang, Anantharam, and Geng developed a
linear-matrix-inequality method for higher entropy derivatives
\cite{ZhangAnantharamGeng2018}.  Wang subsequently proved
\eqref{eq:gcmc} through order five for log-concave inputs in every dimension
by higher-order Otto calculus \cite{Wang2025Otto}.

The de Bruijn identity gives
\[
 -\cH_f'(t)=\frac12 I(P_tf),
 \qquad
 I(g)=\int_{\R^d}|\nabla\log g|^2g\,\dd x,
\]
where the second formula applies to a density $g$ on $\R^d$.
Thus~\eqref{eq:gcmc} is equivalent to complete monotonicity of the Fisher
information profile. Since every positive completely monotone function is
log-convex, GCMC would in particular imply log-convexity of
$t\mapsto I(P_tf)$. This weaker consequence already separates dimension one
from higher dimensions. Ledoux, Nair, and Wang proved Fisher-information
log-convexity in dimension one \cite{LedouxNairWang2021}, whereas a recent
hexagonal construction disproved it in every dimension $d\ge2$
\cite{ZouFanGaoWang2026}. Thus a one-dimensional GCMC counterexample cannot be
detected at the log-convexity level and must instead violate a higher-order
complete-monotonicity sign.

Gu and Sellke recently found such a higher-order obstruction
\cite{GuSellke2026}. Their counterexample is a symmetric probability measure
$\mu$ supported on seventeen points. The entropy profile~\eqref{eq:entropy-profile}
extends to probability measures through the density of $P_t\mu$. In our heat
convention, their certified estimate reads
\[
 \frac{0.36}{2^5}<\cH_\mu^{(5)}(2/3)<\frac{0.37}{2^5}.
\]
This is the smallest possible order of failure for unrestricted
one-dimensional data, by the theorem of Cheng and Geng. Their candidate was
found with GPT-5.5 Pro, and its sign was certified from an exact
one-dimensional integral using SageMath/Arb ball arithmetic. The same paper
also established the existence of a log-concave counterexample by heat
regularization and backward propagation of complete monotonicity.

Theorem~\ref{thm:main} gives a direct analytic construction in the log-concave
class. The $m$th member of the explicit family is paired with the $m$th
derivative, whose signed value is negative for every sufficiently large $m$.
An exact entropy expansion and uniform remainder estimates establish the sign
without a separate numerical certificate.

\noindent\textbf{\textit{Proof strategy.}}\ 
We first remove the Gaussian envelope and study the two-frequency factor on
the circle. Expanding its entropy, we group the surviving monomials by spectral
level. A term at level $j$ carries $B_j^m$, where $B_j=j2^{-j/2}$. The unique
dominant level is $j=3$. The two terms at this level have coefficients $-1/8$
and $1/96$, producing the leading constant $-11/96$
in~\eqref{eq:main-asymptotic}; all other levels form a uniformly bounded
remainder.

We then restore the Gaussian envelope. The exact heat evolution of a
Gaussian-windowed Fourier mode reduces the real-line calculation to the
periodic one, up to Fourier-filtering and time-change errors of
$o(\beta^m)$. The curvature contributed by the Gaussian envelope enforces
strict log-concavity, and analyticity preserves the negative sign for
sufficiently small positive times. Finally, Gaussian tensorization and
additivity of entropy give the higher-dimensional counterexamples.

\section{The Periodic Calculation}
\label{sec:periodic}

For an integrable function on $\T=\R/(2\pi\Z)$, write
\[
 \avgT{F}=\frac1{2\pi}\int_0^{2\pi}F(x)\,\dd x.
\]
The torus heat evolution of the periodic factor in~\eqref{eq:candidate} is
\begin{equation}
 r_m(s,x)=1+a_me^{-s/2}\cos x+b_me^{-2s}\cos2x.
 \label{eq:periodic-profile}
\end{equation}
Set
\begin{equation}
 \cE_m(s)=\avgT{r_m(s,\cdot)\log r_m(s,\cdot)},
 \qquad S_m=(-1)^m\cE_m^{(m)}(0).
 \label{eq:periodic-entropy}
\end{equation}

\begin{lemma}[Periodic asymptotics]
\label{lem:periodic}
As $m\to\infty$,
\begin{equation}
 S_m=-\frac{11}{96}\beta^m+O(1).
 \label{eq:periodic-asymptotic}
\end{equation}
In addition, the absolute sum associated with the expansion of $S_m$ is
$O(\beta^m)$.
\end{lemma}

\begin{proof}
Put
\[
 u_m(s,x)=a_me^{-s/2}\cos x+b_me^{-2s}\cos2x.
\]
For $m\ge5$, one has $\lVert u_m(s,\cdot)\rVert_\infty<1/2$ for $s\ge0$.
The uniformly convergent series
\begin{equation}
 (1+u)\log(1+u)
 =u+\sum_{\ell=2}^{\infty}
 \frac{(-1)^\ell}{\ell(\ell-1)}u^\ell
 \label{eq:entropy-series}
\end{equation}
therefore applies to $r_m=1+u_m$. Since $\avgT{u_m}=0$, expand each power of
$u_m$ and let $p$ and $q$ denote the numbers of first- and second-harmonic
factors, respectively. This gives
\begin{equation}
 S_m=\sum_{\substack{p,q\ge0\\p+q\ge2}}
 \Gamma_{p,q}a_m^pb_m^q
 \left(\frac{p+4q}{2}\right)^m,
 \label{eq:Sm-exact}
\end{equation}
where
\begin{equation}
 \Gamma_{p,q}=
 \frac{(-1)^{p+q}}{(p+q)(p+q-1)}
 \binom{p+q}{p}
 \avgT{\cos^p x\cos^q2x}.
 \label{eq:Gamma}
\end{equation}
Termwise differentiation is legitimate. At total degree
$\ell=p+q$, the differentiated exponential is at most $(2\ell)^m$,
whereas the sum of the remaining absolute weights is bounded by
$(a_m+b_m)^\ell$; hence
\[
 \sum_{\ell\ge2}(2\ell)^m(a_m+b_m)^\ell<\infty.
\]

The change of variables $x\mapsto\pi-x$ shows that the average
in~\eqref{eq:Gamma} vanishes when $p$ is odd. Write $p=2h$ and introduce
the level $j=h+2q$. Then $(p+4q)/2=j$ and
\[
 a_m^{2h}b_m^q=2^{-m(h/2+q)}=2^{-mj/2}.
\]
Thus every nonzero term at level $j$ carries the factor
\begin{equation}
 B_j^m,\qquad B_j=j2^{-j/2}.
 \label{eq:spectral-gain}
\end{equation}
Since
\[
 \frac{B_{j+1}}{B_j}=\frac{j+1}{j\sqrt2},
\]
the sequence increases through $j=3$ and decreases thereafter. In
particular,
\[
 B_1=2^{-1/2},\qquad B_2=B_4=1,
 \qquad B_3=\beta,
\]
and $B_j<1$ for every $j\notin\{2,3,4\}$.

There are two pairs at the unique dominant level $j=3$:
\[
 (p,q)=(2,1),\qquad (p,q)=(6,0).
\]
For the first pair,
\[
 \avgT{\cos^2x\cos2x}=\frac14,
 \qquad
 \Gamma_{2,1}=-\frac1{3\cdot2}\binom32\frac14=-\frac18.
\]
For the second,
\[
 \avgT{\cos^6x}=\frac5{16},
 \qquad
 \Gamma_{6,0}=\frac1{6\cdot5}\frac5{16}=\frac1{96}.
\]
Their sum is
\begin{equation}
 \left(-\frac18+\frac1{96}\right)\beta^m
 =-\frac{11}{96}\beta^m.
 \label{eq:dominant-level}
\end{equation}

It remains to sum the other levels. For fixed $j$, the relation
$h+2q=j$ allows at most $j+1$ pairs, while
\[
 p+q=2j-3q\le2j,
 \qquad
 \binom{p+q}{p}\le4^j.
\]
After discarding the denominator in~\eqref{eq:Gamma} and using that the
trigonometric average has modulus at most one, the total absolute coefficient
at level $j$ is bounded by $D_j=(j+1)4^j$. The finitely many levels
$j\le15$, $j\ne3$, contribute $O(1)$. For $j\ge16$,
\[
 \sum_{j\ge16}D_jB_j^{10}
 =\sum_{j\ge16}(j+1)4^j
 \bigl(j2^{-j/2}\bigr)^{10}<\infty.
\]
Since $B_j<1$ on this range, the same series with exponent $m\ge10$ is
uniformly bounded. This proves~\eqref{eq:periodic-asymptotic}. The identical
bounds without signs, together with the two level-three terms, show that the
absolute sum in~\eqref{eq:Sm-exact} is $O(\beta^m)$.
\end{proof}

\section{Gaussian Localization and Completion of the Proof}
\label{sec:transfer}

We first record positivity and strict log-concavity. Direct Gaussian
integration gives
\begin{equation}
 Z_m=\sqrt{\frac{2\pi}{\kappa_m}}\,\Xi_m,
 \qquad
 \Xi_m=1+a_me^{-1/(2\kappa_m)}+b_me^{-2/\kappa_m}.
 \label{eq:normalization}
\end{equation}
Let $\rho_m(x):=r_m(0,x)=1+a_m\cos x+b_m\cos2x$. For $m\ge5$,
\[
 \rho_m(x)\ge1-a_m-b_m>\frac12,
 \qquad |\rho_m''(x)|\le a_m+4b_m.
\]
Moreover, $8b_m\le a_m$, and consequently
\[
 (\log \rho_m)''
 \le\frac{|\rho_m''|}{\rho_m}
 \le2(a_m+4b_m)\le3a_m.
\]
Since $\kappa_m=8a_m$,
\begin{equation}
 (\log f_m)''=-\kappa_m+(\log \rho_m)''
 \le-5a_m<0.
 \label{eq:log-concavity}
\end{equation}

The interaction of the Gaussian window with the heat flow is exact.

\begin{lemma}[Gaussian-windowed Fourier modes]
\label{lem:heat-mode}
For $\kappa>0$, $n\in\Z$, $t\ge0$, and $d_\kappa(t)=1+\kappa t$,
\begin{equation}
 \begin{aligned}
 P_t\!\left(e^{-\kappa x^2/2}e^{\ii nx}\right)
 &=d_\kappa(t)^{-1/2}
 \exp\!\left(-\frac{\kappa x^2}{2d_\kappa(t)}\right)\\
 &\quad\times
 \exp\!\left(-\frac{n^2t}{2d_\kappa(t)}\right)
 \exp\!\left(\frac{\ii nx}{d_\kappa(t)}\right).
 \end{aligned}
 \label{eq:heat-mode}
\end{equation}
\end{lemma}

\begin{proof}
For $t>0$, convolution with the heat kernel gives
\[
 \frac1{\sqrt{2\pi t}}\int_{\R}
 \exp\!\left(-\frac{\kappa y^2}{2}-\frac{(x-y)^2}{2t}+\ii ny\right)\dd y.
\]
Writing $d=d_\kappa(t)$, the exponent equals
\[
 -\frac{d}{2t}\left(y-\frac{x+\ii nt}{d}\right)^2
 -\frac{\kappa x^2}{2d}-\frac{n^2t}{2d}+\frac{\ii nx}{d}.
\]
The Gaussian integral contributes the factor $d^{-1/2}$ and proves
\eqref{eq:heat-mode}. The formula extends to $t=0$ by continuity.
\end{proof}

Specialize the lemma to~\eqref{eq:candidate} and set
\begin{equation}
 d_m(t)=1+\kappa_mt,\qquad
 \tau_m(t)=\frac{t}{d_m(t)},\qquad
 \lambda_m(t)=\kappa_md_m(t).
 \label{eq:time-functions}
\end{equation}
With $r_m(s,y)$ as in~\eqref{eq:periodic-profile},
\begin{equation}
 P_tf_m(x)=Z_m^{-1}d_m(t)^{-1/2}
 e^{-\kappa_mx^2/(2d_m(t))}
 r_m\!\left(\tau_m(t),\frac{x}{d_m(t)}\right).
 \label{eq:heat-factorization}
\end{equation}
For $\lambda>0$, let $\gamma_\lambda$ be the centered Gaussian law of
precision $\lambda$ and write
\[
 \dd\gamma_\lambda(y)=\sqrt{\frac{\lambda}{2\pi}}
 e^{-\lambda y^2/2}\,\dd y,
 \qquad
 \avgG{F}{\lambda}=\int F\,\dd\gamma_\lambda.
\]
After the change of variables $y=x/d_m(t)$,~\eqref{eq:heat-factorization}
one first obtains
\begin{equation}
 (P_tf_m)(x)\,\dd x
 =\Xi_m^{-1}r_m(\tau_m(t),y)\,\dd\gamma_{\lambda_m(t)}(y).
 \label{eq:measure-factorization}
\end{equation}
Taking the logarithm in~\eqref{eq:heat-factorization} and using
\eqref{eq:measure-factorization} gives
\begin{equation}
 \begin{aligned}
 \cH_{f_m}(t)
 &=-\log Z_m-\frac12\log d_m(t)\\
 &\quad+\Xi_m^{-1}\left[
 \avgG{G_m(t,\cdot)}{\lambda_m(t)}+M_m(t)
 \right],
 \end{aligned}
 \label{eq:entropy-decomposition}
\end{equation}
where
\begin{equation}
 \begin{aligned}
 G_m(t,y)&=r_m(\tau_m(t),y)\log r_m(\tau_m(t),y),\\
 M_m(t)&=-\frac{\lambda_m(t)}2
 \avgG{y^2r_m(\tau_m(t),y)}{\lambda_m(t)}.
 \end{aligned}
 \label{eq:G-and-M}
\end{equation}

\begin{lemma}[Gaussian transfer]
\label{lem:transfer}
As $m\to\infty$,
\begin{equation}
 (-1)^m\cH_{f_m}^{(m)}(0)=S_m+o(\beta^m).
 \label{eq:transfer}
\end{equation}
\end{lemma}

\begin{proof}
We estimate the four nonconstant terms in~\eqref{eq:entropy-decomposition}.
Throughout the proof, $c,C>0$ denote constants independent of $m$, $n$,
and $t$; their values may change from line to line.

\textbf{\textit{Fourier filtering.}}
Fix the strip $|\operatorname{Im}z|\le1$ and the time disk $|t|\le1/2$.
For all sufficiently large $m$, $d_m(t)$ stays away from zero,
$|\tau_m(t)|\le1$, and
\[
 \left|a_me^{-\tau_m(t)/2}\cos z\right|
 +\left|b_me^{-2\tau_m(t)}\cos2z\right|<\frac12.
\]
Thus the principal logarithm defining $G_m(t,z)$ is analytic and uniformly
bounded on this product domain. Its Fourier series
\[
 G_m(t,y)=\sum_{n\in\Z}\widehat G_m(n,t)e^{\ii ny}
\]
converges normally on the real axis, and the strip estimate gives
\begin{equation}
 |\widehat G_m(n,t)|\le Ce^{-|n|}
 \label{eq:Fourier-decay}
\end{equation}
uniformly on the time disk.

For complex $\lambda$ with $\operatorname{Re}\lambda>0$, we use the same
bracket notation for the holomorphic extension of the Gaussian integral,
defined with the principal square root. Direct integration gives
\[
 \sqrt{\frac{\lambda}{2\pi}}
 \int_{\R}e^{\ii ny}e^{-\lambda y^2/2}\,\dd y
 =e^{-n^2/(2\lambda)}.
\]
Normal convergence permits termwise integration. Since
\[
 \operatorname{Re}\frac1{\lambda_m(t)}
 =\frac1{\kappa_m}
 \operatorname{Re}\frac1{1+\kappa_mt}
 \ge\frac{c}{\kappa_m},
\]
the nonzero Fourier modes satisfy, uniformly for $|t|\le1/2$,
\begin{equation}
 \left|
 \avgG{G_m(t,\cdot)}{\lambda_m(t)}
 -\avgT{G_m(t,\cdot)}
 \right|
 \le Ce^{-c/\kappa_m}.
 \label{eq:filtering-bound}
\end{equation}
Cauchy's estimate therefore bounds the $m$th derivative at zero by
$C2^m m!e^{-c/\kappa_m}$. Since
$\kappa_m=8\,2^{-m/4}$, this quantity is $o(\beta^m)$. The zero Fourier
coefficient in~\eqref{eq:filtering-bound} is
\begin{equation}
 \avgT{G_m(t,\cdot)}=\cE_m(\tau_m(t)).
 \label{eq:zero-mode}
\end{equation}

\textbf{\textit{Quadratic moment.}}
The Gaussian characteristic function yields
\[
 \avgG{y^2e^{\ii ny}}{\lambda}
 =\left(\frac1\lambda-\frac{n^2}{\lambda^2}\right)
 e^{-n^2/(2\lambda)}.
\]
The identity
\begin{equation}
 \tau_m(t)+\lambda_m(t)^{-1}=\kappa_m^{-1}
 \label{eq:conservation}
\end{equation}
then gives
\begin{equation}
 \begin{aligned}
 M_m(t)&=-\frac12-\frac{a_m}{2}e^{-1/(2\kappa_m)}
 +\frac{a_m}{2\lambda_m(t)}e^{-1/(2\kappa_m)}\\
 &\quad-\frac{b_m}{2}e^{-2/\kappa_m}
 +\frac{2b_m}{\lambda_m(t)}e^{-2/\kappa_m}.
 \end{aligned}
 \label{eq:moment-explicit}
\end{equation}
Only $\lambda_m(t)^{-1}=\kappa_m^{-1}(1+\kappa_mt)^{-1}$ has a
nonconstant $m$th derivative. Hence
\begin{equation}
 \begin{aligned}
 (-1)^mM_m^{(m)}(0)
 &=m!\kappa_m^{m-1}
 \left[\frac{a_m}{2}e^{-1/(2\kappa_m)}
 +2b_me^{-2/\kappa_m}\right]\\
 &=o(\beta^m).
 \end{aligned}
 \label{eq:moment-small}
\end{equation}

\textbf{\textit{Gaussian entropy.}}
The signed derivative of the second term in~\eqref{eq:entropy-decomposition}
is
\begin{equation}
 (-1)^m\frac{\dd^m}{\dd t^m}
 \left[-\frac12\log(1+\kappa_mt)\right]_{t=0}
 =\frac12(m-1)!\kappa_m^m.
 \label{eq:gaussian-small}
\end{equation}
Its logarithm equals $-(\log2)m^2/4+O(m\log m)$, so it is
$o(\beta^m)$.

\textbf{\textit{Time change.}}
For $L>0$ and $\kappa\ge0$, coefficient extraction from
$e^{-Lt/(1+\kappa t)}$ gives
\begin{equation}
 \begin{aligned}
 T_m(L,\kappa)
 &:=(-1)^m\frac{\dd^m}{\dd t^m}
 e^{-Lt/(1+\kappa t)}\Big|_{t=0}\\
 &=\sum_{r=0}^{m-1}
 \frac{m!}{(m-r)!}\binom{m-1}{r}
 \kappa^rL^{m-r}.
 \end{aligned}
 \label{eq:time-change-formula}
\end{equation}
Indeed, expand first in powers of $Lt/(1+\kappa t)$ and then use
$(1+\kappa t)^{-k}=\sum_{v\ge0}(-1)^v
\binom{k+v-1}{v}\kappa^vt^v$; the coefficient of $t^m$ yields
\eqref{eq:time-change-formula}. The estimates
\[
 \frac{m!}{(m-r)!}\le m^r,
 \qquad
 \binom{m-1}{r}\le\frac{m^r}{r!}
\]
imply
\begin{equation}
 0\le\frac{T_m(L,\kappa)}{L^m}-1
 \le e^{m^2\kappa/L}-1.
 \label{eq:time-change-bound}
\end{equation}
Every exponential in the periodic series~\eqref{eq:Sm-exact} has
$L=(p+4q)/2\ge1$. Since $m^2\kappa_m\to0$, the relative error
in~\eqref{eq:time-change-bound} is $o(1)$ uniformly over all terms. The
absolute-sum estimate in Lemma~\ref{lem:periodic} therefore gives
\begin{equation}
 (-1)^m\frac{\dd^m}{\dd t^m}
 \cE_m(\tau_m(t))\Big|_{t=0}=S_m+o(\beta^m).
 \label{eq:time-change-entropy}
\end{equation}

Finally,~\eqref{eq:normalization} gives
$\Xi_m=1+O(e^{-c/\kappa_m})$. Combining
\eqref{eq:filtering-bound}--\eqref{eq:time-change-entropy} with
\eqref{eq:entropy-decomposition} proves~\eqref{eq:transfer}.
\end{proof}

\begin{proof}[Proof of Theorem~\ref{thm:main}]
The estimates preceding Lemma~\ref{lem:heat-mode} prove smoothness,
positivity, Gaussian decay, and strict log-concavity. Lemmas
\ref{lem:periodic} and~\ref{lem:transfer} give
\[
 (-1)^m\cH_{f_m}^{(m)}(0)
 =-\frac{11}{96}\beta^m+o(\beta^m),
\]
which is negative for all sufficiently large $m$.

For each such fixed $m$, the exact representation
\eqref{eq:entropy-decomposition} extends holomorphically to a neighborhood
of zero. Indeed, the principal logarithm used in the filtering argument is
uniformly analytic there, the spatial Fourier series converges normally, and
$d_m(t)$ stays away from zero. Thus $\cH_{f_m}^{(m)}(t)$ is right-continuous
at zero, and the strict negative sign persists on some interval
$[0,\varepsilon_m)$. This proves~\eqref{eq:positive-time-failure}.
\end{proof}

A related broad-Gaussian tensorization appears in
\cite{ZouFanGaoWang2026}; we record the entropy calculation here.

\begin{corollary}[Higher-dimensional counterexamples]
\label{cor:high-dimensional}
For integers $m\ge5$ and $d\ge2$, put $r=d-1$, let
\[
 g_v(y)=(2\pi v)^{-1/2}e^{-y^2/(2v)},
 \qquad F_{m,d}=f_m\otimes g_{rm}^{\otimes r},
\]
and interpret \eqref{eq:entropy-profile} using the heat semigroup
$e^{t\Delta/2}$ on $\R^d$. Then $F_{m,d}$ is smooth, strictly positive,
Gaussian-decaying, and strictly log-concave. As $m\to\infty$, uniformly over
$d\ge2$,
\begin{equation}
 (-1)^m\cH_{F_{m,d}}^{(m)}(0)
 =-\frac{11}{96}\beta^m+o(\beta^m).
 \label{eq:high-dimensional-asymptotic}
\end{equation}
Consequently, for every sufficiently large $m$ and every $d\ge2$, the GCMC
fails for a smooth strictly log-concave density on $\R^d$.
\end{corollary}

\begin{proof}
The heat semigroup tensorizes and $P_tg_v=g_{v+t}$. Additivity of Boltzmann
entropy gives
\[
 \cH_{F_{m,d}}(t)
 =\cH_{f_m}(t)-\frac{r}{2}\log\!\bigl(2\pi e(rm+t)\bigr).
\]
Consequently,
\begin{align*}
 (-1)^m\cH_{F_{m,d}}^{(m)}(0)
 &=(-1)^m\cH_{f_m}^{(m)}(0)
 +\frac{r}{2}(m-1)!(rm)^{-m},\\
 0&\le\frac12r^{1-m}\frac{(m-1)!}{m^m}\le\frac{1}{2m}.
\end{align*}
Theorem~\ref{thm:main} proves
\eqref{eq:high-dimensional-asymptotic}; right-continuity transfers its strict
negative sign to a positive-time interval. Finally,
$\nabla^2\log F_{m,d}=\operatorname{diag}((\log f_m)'',-(rm)^{-1}I_r)\prec0$;
the other stated properties also follow directly from the product formula.
\end{proof}

\finalcolumnbalance
\noindent\textbf{Disclosure on the use of generative AI.}
The authors proposed a Gaussian-windowed two-frequency ansatz
$\widetilde Z^{-1}e^{-\lambda x^2/2}
\bigl(c+\alpha\cos x+\eta\cos2x\bigr)$ and the strategy of adapting their
earlier torus-to-Euclidean localization. Through multiple interactions,
\mbox{GPT-5.6 Sol Pro} identified the parameter scaling, and assisted with its
exposition. The authors verified and revised the
argument and take full responsibility for all mathematical claims.
 

\end{document}